\documentclass[12pt]{article}
\usepackage{amsmath}
\usepackage{amsfonts}
\usepackage{amssymb}
\usepackage{latexsym}
\usepackage[dvips]{graphicx}

\newcommand{\cta}{\theta_a}

\newtheorem{theorem}{\bf Theorem}[section]
\newtheorem{lemma}[theorem]{\bf Lemma}
\newtheorem{prop}[theorem]{\bf Proposition}
\newtheorem{coro}[theorem]{\bf Corollary}

\newenvironment{proof}{\noindent{\em Proof:}}{\quad \hfill$\Box$\vspace{2ex}}

  \def\hml{\end{document}}  \newsymbol\wjzhml 203F 

\newcommand{\dd}[1]{\, {\mathrm d}{#1}}
\newcommand{\sinc}{\mathrm{sinc}}
\newcommand{\e}{\mathrm{e}}
\newcommand{\im}{\mathrm{i}}

\begin{document}

\title{ Amplitudes of mono-components and representation by generalized sampling functions}

\author{ Qiuhui Chen \thanks{ Qiuhui Chen, Cisco School of Informatics, Guangdong University of
Foreign Studies, Guangzhou,  CHINA.  Email: {chenqiuhui@hotmail.com}},
\quad Luoqing Li \thanks{Luoqing Li, Faculty of Mathematics and Computer Science, Hubei University, Wuhan, 430062 CHINA. Email: {lilq@hubu.edu.cn}},
\quad Yi Wang \thanks{Yi Wang, corresponding author, Department of Mathematics, Auburn University at Montgomery,
 P.O. Box 244023, Montgomery, AL 36124-4023
 USA.    Email: {ywang2@aum.edu}
 }
 }

\date{}

\maketitle

\begin{abstract}
A mono-component is a real-valued signal of finite energy that has non-negative instantaneous frequencies, which may be defined as  the derivative of the phase function  of  the given real-valued signal through the approach of canonical amplitude-phase modulation. We study in this article how the amplitude is determined by its phase in a canonical amplitude-phase modulation. Our finding is that such an amplitude can  be perfectly reconstructed by a sampling formula using  the so-called generalized sampling functions and their Hilbert transforms. The regularity of such an amplitude is identified to be at least continuous. Meanwhile, we also make a very interesting and new characterization of the band-limited functions.
\end{abstract}

\section{Introduction}
\setcounter{equation}{0}

Any real-valued non-stationary signal $f$ of finite energy, that is, $f$ is in the space $ L^2(\mathbb{R})$ of square integrable functions on the set $\mathbb{R}$ of real numbers,    may be represented as an
{\em amplitude-phase modulation} with a time-varying amplitude $\rho
$ and a time-varying phase $\phi$
where   phase $\phi  $ is, in general, {\em nonlinear}. Specifically,   the value of $f$ at $t\in\mathbb R$ may be represented as
\begin{equation}\label{amplitude-frequency-modulated-signal}
f(t)=\rho (t)\cos\phi(t).
\end{equation}
Unfortunately, this type of representation is not unique because the modulation is obtained through a complex signal that can have various choices of the imaginary part. However,
one can determine a unique such factorization
(\ref{amplitude-frequency-modulated-signal}) by using the approach of  analytic signals. Indeed, let $\mathcal{A}(f)$ be the {\it analytic signal
associated with} $f$ with the characteristic property
\begin{equation}
 ({\mathcal{A}(f)})^\wedge(\omega)=\left\{\begin{array}{ll}2{\hat
f}(\omega)& ~\mbox{if} ~ \omega\ge 0\\ 0 &~\mbox{if}~ \omega< 0,
\end{array} \right. \label{eqn:analy}
\end{equation}
where for any signal $g \in L^2(\mathbb{R})$, $ \hat{g}={\mathcal F}g $ is the Fourier transform of $g$ defined at $\xi\in \mathbb{R}$ by the equation
\begin{equation}\label{eqn:fourier}
\hat{g}(\xi)=({\mathcal F}g)(\xi):=\frac{1}{\sqrt{2\pi}}\int_{\mathbb R}
g(t)\e^{-\im\xi t}\dd{t} .
\end{equation}
Equation \eqref{eqn:analy} is equivalent to   for $t\in \mathbb{R}$,
\begin{equation*}
\mathcal{A}(f)(t)=f(t)+\im{\mathcal H}f(t), \label{analytic=f+if}
\end{equation*}
where  the operator ${\mathcal H}:L^2(\mathbb{R})\to L^2(\mathbb{R})$ stands for the Hilbert transform, and for $f\in L^2(\mathbb{R})$, $\mathcal{H}f$ at $t\in \mathbb{R}$ is  defined through
the principal value integral
\begin{equation*}
{\mathcal H}f(t):={\mathrm{p.v.}}\frac{1}{\pi}\int_{\mathbb{R}}
\frac{f(x)}{t-x}\dd{x}=\lim_{\epsilon\to 0}\int_{|x-t|>\epsilon
}\frac{f(x)}{t-x}\dd{x}. \label{hilbert-transform}
\end{equation*}
The  value   $\mathcal{A}(f)(t) $ at $t\in \mathbb{R}$ is complex  which can
be written into the {\em quadrature form}
\begin{equation*}
\mathcal{A}(f)(t)=\rho (t)\e^{\im\phi(t)}.
\label{analytic-complex-form}
\end{equation*}

 Under the conditions that the derivative value $\phi'(t)$ is non-negative or non-positive
for all $t\in \mathbb{R}$,   the quantities $\rho (t)$ and $\phi^\prime(t)$ are
called the {\it instantaneous amplitude} and {\it instantaneous
frequency} at $t\in \mathbb{R}$, of the real signal $f$, respectively.  The corresponding modulation
(\ref{amplitude-frequency-modulated-signal}) is then called the {\em
canonical amplitude-phase modulation,} or {\em canonical modulation}
for short. The  signal $f$ with such defined non-negative  instantaneous frequencies
is  thus called a {\em mono-component}. A large body of literature addresses this
problem, see for example, \cite{Bedrosian,Boashash, Cohen,Nuttall, Picinbono, Doro2003, XCo1999}.

Constructing the  canonical pair $(\rho,\phi)$ of
the instantaneous amplitude and phase  is important in
the theories of analytic signals.
It is equivalent to  the problem of seeking the function pair   $(\rho, \phi)$
  such that for $t\in \mathbb{R}$, the following equation holds true
\begin{equation}\label{eq:bedrosian-phase-Qian}
{\cal
H}\left(\rho(\cdot)\cos\phi(\cdot)\right)(t)=\rho(t)\sin\phi(t).
\end{equation}
We remark that (\ref{eq:bedrosian-phase-Qian}) can be apparently
considered as a special case of the {\em Bedrosian identity}
$$
{\mathcal H}(fg)=f{\mathcal H}(g).
$$
In \cite{Bedrosian}, the author proved that, if both $f,g$ belong to
$L^2(\mathbb R)$, $f$ is of lower frequency, $g$ is of higher
frequency and $f,g$ have no overlapping frequency, then ${\mathcal
H}(fg)=f{\mathcal H}(g)$.  This classic
result of Bedrosian is not useful for constructing a
mono-component. The reason lies in that the requirement of both
$f$ and $g$  in $L^2(\mathbb R)$ is invalid.

Recently,   an  important   phase function  that renders mono-components  was given
 in \cite{QCL_PhyD05}. The phase function is defined through the boundary values of a
    Blaschke product on a unit disk $\Delta:=\{z:z\in \mathbb{C}, |z|\le 1  \}   $, where $\mathbb{C}$ indicates the set of complex numbers.   Specifically, for
$a \in (-1,1) $, the Blaschke product
  at $z\in\mathbb C\backslash\{1/ {a}\} $ is given by
\begin{equation}\label{eqn:Blaschke-product}
B_{  a}(z)=
\frac{z-a}{1-{a}z}.
\end{equation}
Subsequently the
  {\em non-linear phase} function, denoted by $\theta_a$, is defined at $t\in \mathbb{R}$ by  the equation
\begin{equation}\label{eqn:phase}
 \e^{\theta_{  a}(t)} :=B_{  a}(\e^{\im t}).
\end{equation}
If we recall that   the periodic Poisson kernel  $p_a$   whose value at $t\in \mathbb{R}$ is given by
\begin{equation}\label{eqn:pa}
 p_a(t):=\frac{1-a^2}{1-2a\cos t+a^2},
 \end{equation}
then by taking the derivative of both sides of equation \eqref{eqn:phase}, we find  that
the phase $\theta_a$ is an anti-derivative of $p_a$, and its derivative is always positive,   that is, $$\frac{\dd{}}{\dd{t}}\theta_a(t)= p_a(t)>0.$$

We shall in this paper characterize the amplitude function $\rho$ of finite energy when the phase function $\phi$ is chosen     at $t\in \mathbb{R}$    by $$\phi(t)=\theta_a(t)=\displaystyle\int_{[0,t]}p_a(x)\dd{x}$$ such that  equation (\ref{eq:bedrosian-phase-Qian}) is satisfied. Our main result indicates that such kind of amplitude can be perfectly reconstructed  in terms of a   sampling formula using the {\em  generalized sampling function} whose value at $t\in \mathbb{R}$ is given by
\begin{equation}\label{eq-sinc-one-parameter}
\sinc_{  a}( t):=\frac{\sin \theta_{
a}( t)}{ t}.
\end{equation}

  In Section 2, we review the construction of the generalized sampling function and discuss some properties pertained to  it. In Section 3, we introduce the concept of Bedrosian subspace of the Hilbert transform and investigate some properties of functions in this space.  In Section 4, we make an important observation when a {\em linear phase} is chosen, the amplitude function must be bandlimited in order to satisfy equation \eqref{eq:bedrosian-phase-Qian}. In Section 5, we present our main result in Theorem \ref{thm:sampling}.

\section{Generalized sampling functions}
Not very surprisingly the function $\sinc_a$ has many properties that are similar to the classic {\em sinc} defined at   $t\in \mathbb{R}$    by the equation
 $$
\mathrm{sinc} (t):=\frac{\sin t}{t}. $$ Those properties include  cardinality, orthogonality, decaying rate, among others. In the special case $a=0$, the function $\sinc_a$   reduces to the classic $\sinc$, which will become clear later. Let us first review the approach to obtain an explicit form of $\sinc_a$.

The classic  {sinc} function
is fundamentally significant in digital signal processing due to
the Shannon
 sampling theorem
\cite{Sha1948,Sha1949,BER1986}.
The Shannon sampling theorem enables to reconstruct a bandlimited signal from shifts of sinc functions weighted by the uniformly spaced samples of that signal. Recently efforts have been made to extend the classic sinc to  {generalized sampling functions}, for example, in \cite{CMW2010,ChenQian-AA09,Chen-Wang-Wang}.   Intuitively, the spectrum of the
sinc function is just the indicator function of a symmetric
interval of finite measure. Hence, authors in \cite{CMW2010} are inspired to consider functions with
piecewise polynomial spectra to replace the usual sinc function for
the purpose of sampling   non-bandlimited signals.  One kind of  {generalized sampling functions} given in \cite{CMW2010}, denoted by
$\sinc_a$ that is related to a constant $a\in (-1,1)$, is defined as the {\em inverse Fourier transform} of a so-called {\em symmetric cascade filter}, denoted by  $H_a $. Specifically,
\begin{equation}\label{general-sinc-function}
\sinc_a:=\sqrt{\frac{\pi}{2}}(1+a){\mathcal F}^{-1}H_a .
\end{equation}

Let $\mathbb{N}$ be the set of natural numbers, $\mathbb{Z}$ be the set of integers, and $\mathbb{Z}_+:=\{0\}\cup \mathbb{N}$.  Let $X$ be a subset of $\mathbb{R}$, and for $q\in \mathbb{N}$,  we say a function $f$ is in $L^q(X)$ if and only if the $L^q(X)$ norm
$$
\|f\|_{q,X}:=\left(\int_X |f(t)|^q \dd{t}\right)^{1/q}<\infty.
$$
Similarly, let $Z$ be a subset of $\mathbb{Z}$, a sequence ${\boldsymbol y}:=(y_k:k\in Z) $ is said to be in $l^q(Z)$ if and only if the $l^q(Z)$ norm
$$
\|{\boldsymbol y}\|_{q,Z}:=\left(\sum_{k\in Z} |y_k|^q\right)^{1/q}<\infty.
$$

The symmetric cascade filter $H_a$ is a  piecewise constant function whose value at $\xi\in \mathbb{R}$ is given by
\begin{equation}\label{eqn:general-H}
H_a (\xi):=\sum_{n\in \mathbb{Z}_+}a^n \chi_{I_n}(\xi),
\end{equation}
where   $\chi_I$ is the {\em indicator function} of the set $I$, and the interval $I_n$, $n\in \mathbb{Z}_+$, is the union of two symmetric intervals  given by the equation
$$
{I}_{n} := (-(n+1) ,-n ]\cup[n ,(n+1) ).
$$
Of course, we have that $H\in L^1(\mathbb{R})\cap L^2(\mathbb{R})$ because the sequence ${\boldsymbol
b}:=(a^n: n\in \mathbb{Z}_+) \in l^1({\mathbb Z}_+)\cap l^2({\mathbb Z}_+) $, and hence $\sinc_a\in L^2(\mathbb{R})$ since the Fourier operator is closed in  $L^2(\mathbb{R})$ and $\sinc_a$ is continuous because $H_a\in L^1(\mathbb{R})$.

The symmetric cascade filter $H_a $ can be associated with an {\em analytic function} $F$  on the open unit disk  $\Delta$ defined at   $z \in \Delta$     by
\begin{equation}\label{analfu}
F(z):= \sum_{n\in\mathbb{Z_+}}a^n z^n.
\end{equation}
Thus, by substituting  equation  \eqref{eqn:general-H} into equation \eqref{general-sinc-function} and making use of equation \eqref{analfu} an alternative form of
$\sinc_a (t)$,  $t\in\mathbb R$ in terms of $F$ can be found as
\begin{equation}\label{phi-explicit-form}
\sinc_a(t)=(1+a)\mathrm{sinc}\left(\frac{  t}{ 2}\right)
\mathrm{Re}\left\{F(\e^{\im t})\e^{\frac{1}{2}\im  t} \right\} ,
\end{equation}
where $\mathrm{Re}(z)$ is the real part of a complex number $z$.

A very interesting fact, as discovered in the paper \cite{CMW2010}, is that  the function $F$ is linked to the Blashke Product $B_a$ by the equation
\begin{equation}\label{eqn:def_F}
(1+a)F(z):=\frac{B_a(z)-1}{z-1}.
\end{equation}
Plugging the formula  \eqref{eqn:def_F} into equation \eqref{phi-explicit-form} we readily obtain  the explicit expression of $\sinc_a$ given earlier in equation \eqref{eq-sinc-one-parameter}.

The next two formulas shall be used later. Expanding the left-hand side of equation \eqref{eqn:phase} using  Euler's formula and separating the real part from the imaginary part of the  right-hand side, one obtains that for $t\in \mathbb{R}$
\begin{equation}\label{eqn:sina}
\sin \theta_a(t)=\frac{(1-a^2)\sin t }{1-2a \cos t +a^2}=p_a(t) \sin t
\end{equation}
and
 \begin{equation}\label{eqn:costhetaa}
\cos \theta_a(t)=\frac{(1+a^2)\cos t -2a }{1-2a \cos t +a^2}.
\end{equation}

We next list some properties of the function $\sinc_a$.
\begin{prop}\label{prop:sinc}
Let  the generalized sampling function $\sinc_a$ be defined by equation \eqref{phi-explicit-form} or equation \eqref{eq-sinc-one-parameter}.   Then the following statements hold.
\begin{enumerate}
\item For $t\in \mathbb{R}$, \begin{equation}\label{eqn:sinca}
\sinc_a(t)=\frac{(1-a^2)  }{1-2a \cos t +a^2}\frac{\sin t}{t}=p_a(t) \sinc (t).\end{equation}
\item \begin{equation}\label{general-sinc-function-f}
\mathcal{F}\sinc_a=\sqrt{\frac{\pi}{2}}(1+a) H_a .
\end{equation}
\item   $\sinc_a(n\pi)=\frac{1+a}{1-a}\delta_{n,0}$, where $\delta_{n,0}=1$ if $n=0$ and $\delta_{n,0}=0$ if $n\in \mathbb{Z}\setminus \{0\}$.
    \item $\sinc_a$ is even, bounded, infinitely differentiable.

  \item $|\sinc_a(t)|\le \frac{1+a}{1-a}\frac{2}{1+|t|}$ for $t\in \mathbb{R}$, and $ \sinc_a \in L^2(\mathbb{R})$.

    \item The set $\{ \sinc_a(\cdot -n\pi): n\in \mathbb{Z}  \}$ is an orthogonal set, that is
        $$
        \langle \sinc_a, \sinc_a(\cdot -n\pi)  \rangle = \frac{(1+a) }{1-a} \pi \delta_{n,0},
        $$
        where $\langle u, v \rangle =\int_{\mathbb{R}} u(t){v}^*(t)\dd{t}$ denotes the usual inner product   of the two functions $u,v\in L^2(\mathbb{R})$, and ${v}^*$ is the complex conjugate of $v$.
\end{enumerate}
\end{prop}

\begin{proof}
Equation \eqref{eqn:sinca} directly follows by substituting equation \eqref{eqn:sina}  into equation \eqref{eq-sinc-one-parameter}. Equation \eqref{general-sinc-function-f} is obtained by taking the inverse Fourier transform of both sides of equation \eqref{general-sinc-function}.
The third and fourth statements directly follow from  equation \eqref{eqn:sinca}.
 The fifth statement follows from equation   \eqref{eqn:sinca} and noticing $  \sinc (t) \le \frac{2}{1+|t|}$ and  $p_a(t)\le \frac{1+a}{1-a}$ for any $t\in \mathbb{R}$.
The last statement is a special case of Corollary 3.2 of \cite{CMW2010}. For the convenience of readers, we provide a direct proof here. By Parseval's theorem and    equation  \eqref{general-sinc-function-f}  we have
\begin{eqnarray*}
\int_\mathbb{R} \sinc_a(t)\sinc_a(t-n \pi)\dd{t} &=&  \frac{\pi}{2}(1+a)^2\int_\mathbb{R} H_a^2(x) \e ^{\im n\pi x}\dd{x}\\
&=& \frac{\pi}{2}(1+a)^2\int_\mathbb{R} \sum_{k\in \mathbb{Z}_+}a^{2k} \chi_{I_k}(x) \e ^{\im n\pi x}\dd{x}\\
&=& \frac{\pi}{2}(1+a)^2\sum_{k\in \mathbb{Z}_+}a^{2k}  \int_{I_k} \e ^{\im n\pi x}\dd{x}
= \frac{1+a}{1-a}\pi   \delta_{n,0},
\end{eqnarray*}
where, in the last equality we have used the orthogonality  identity $\int_{I_k}\e ^{\im n\pi x}\dd{x}=2\delta_{n,0} $,    $k\in \mathbb{Z}_+$.
The interchange of the integral operator and the infinite sum is guaranteed by the absolute convergence of the series.
\end{proof}

In figure \ref{fig:fourierpair}, we show an example of the Fourier transform pair $\sinc_a$ and $\sqrt{\frac{\pi}{2}}(1+a)H_a$ with $a=0.5$. In the plot of $\sinc_a$, the graph of the standard $\sinc$ is also shown, which corresponds to the case $a=0$.
\begin{figure}[ht]
$$
\includegraphics[scale=0.6]{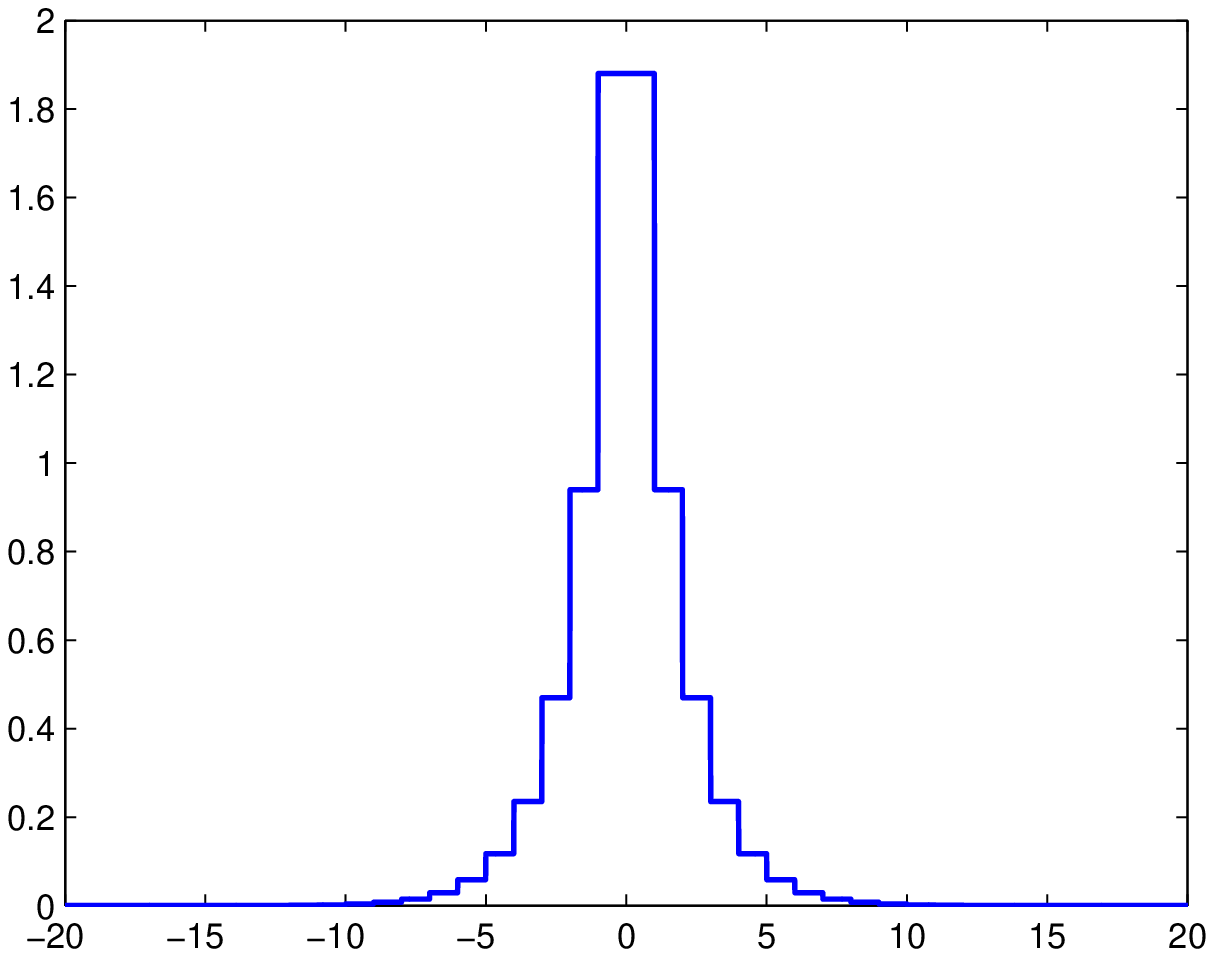}\includegraphics[scale=0.6]{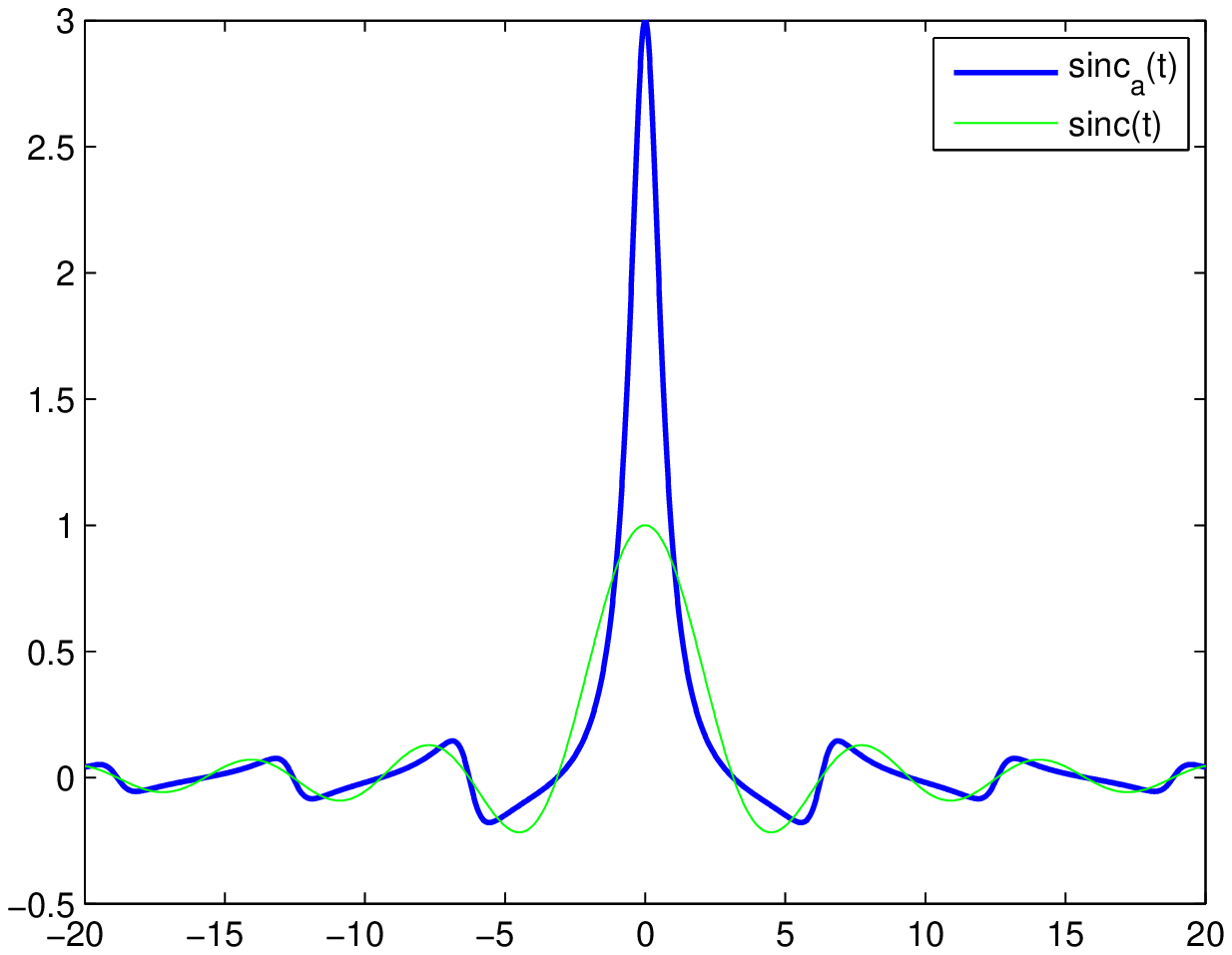}
$$
\caption{Left: the graph of $\sqrt{\frac{\pi}{2}}(1+a)H_a$, $a=0.5$; Right: the graph of $\sinc_a$,   $a=0.5$.}\label{fig:fourierpair}
\end{figure}

Next we review several basic properties of the Hilbert transform which we will need  frequently later.  These properties can be found, for example, in the book \cite{King2009}. First the Hilbert transform is an {\em anti-involution}, that is,
\begin{equation}\label{eqn:Hil1}
\mathcal{H}^2=-\mathcal{I}
\end{equation}
where $\mathcal{I}$ is the identity operator. Second the operator $\mathcal{H}$ is anti-self adjoint, that is,
\begin{equation}\label{eqn:Hil2}
\langle \mathcal{H} u, v  \rangle = \langle u, -\mathcal{H}v   \rangle.
\end{equation}
Third, for any $f\in L^2(\mathbb{R})$ and $t\in \mathbb{R}$, the composition of  the Fourier transform and the Hilbert transform is given by
\begin{equation}\label{eqn:Hil3}
\mathcal{F}(\mathcal{H} f )(t)=-\im \ \mathrm{sgn}(t)\mathcal{F}f(t),
\end{equation}
where $\mathrm{sgn}(\cdot)$ is the {\em signum function} having values defined by  $\mathrm{sgn}(x)=1$ if $x\in \mathbb{R}_+:=\{t\in \mathbb{R}:t>0\}$, $\mathrm{sgn}(x)=-1$ if $x\in \mathbb{R}_-:=\{t\in \mathbb{R}:t<0\}$, and $\mathrm{sgn}(0)=0$.

\begin{theorem}\label{prop:ortho}
The system
\begin{equation}\label{eqn:Phi}
\Phi:=\{ \sinc_{a} (\cdot -2k\pi), \mathcal{H}
\sinc_{a} (\cdot-2k\pi):k\in \mathbb{Z} \}
\end{equation}
 is an orthogonal
system in $L^2(\mathbb{R})$.
\end{theorem}
\begin{proof}
By the third statement  of Proposition \ref{prop:sinc}, we have that
\begin{eqnarray*}
\langle \sinc_a, \sinc_a(\cdot-2k\pi)\rangle = \frac{1+a}{1-a}\delta_{0k}.
\end{eqnarray*}
Invoking equations \eqref{eqn:Hil1} and \eqref{eqn:Hil2} yields
\begin{eqnarray*}
\langle \mathcal{H}\sinc_a, \mathcal{H}\sinc_a(\cdot-2k\pi)\rangle &=&\langle  \sinc_a, -\mathcal{H}^2\sinc_a(\cdot-2k\pi)\rangle\\
&=&\langle \sinc_a, \sinc_a(\cdot-2k\pi)\rangle= \frac{1+a}{1-a}\delta_{0k}.
\end{eqnarray*}
  By Parseval's theorem and equation \eqref{general-sinc-function-f} we have
\begin{eqnarray}
\int_\mathbb{R} \sinc_a(t) \mathcal{H}\sinc_a(t-2k\pi)\dd{t}&=& \frac{(1+a)^2\pi \im }{2}\int_\mathbb{R} H_a^2(t)\mathrm{sgn} (t) \e^{\im 2k\pi t }\dd{t}. \label{eqn:product}
\end{eqnarray}
Invoking the expression \eqref{eqn:general-H} of $H_a$,
equation \eqref{eqn:product} becomes
\begin{eqnarray*}
&&\int_\mathbb{R} \sinc_a(t) \mathcal{H}\sinc_a(t-2k\pi)\dd{t}\\
&=& \frac{(1+a)^2\pi \im }{2} \int_{\mathbb{R}} \sum_{n\in \mathbb{Z}_+}a^{2n} \chi_{I_n}(t)\mathrm{sgn} (t) \e^{\im 2k\pi t }\dd{t}\\
&=& \frac{(1+a)^2\pi \im }{2} \sum_{n\in \mathbb{Z}_+}a^{2n}\int_{I_n} \mathrm{sgn} (t) \e^{\im 2k\pi t }\dd{t}\\
&=& \frac{(1+a)^2\pi \im }{2} \sum_{n\in \mathbb{Z}_+} a^{2n}\left[     \int_{[n,n+1)} \e^{\im 2k\pi t }\dd{t} - \int_{(-(n+1),-n]} \e^{\im 2k\pi t }\dd{t}   \right],
\end{eqnarray*}
where the interchange of the integral and the infinite sum in the second equality is guaranteed by the absolute convergence of the series.

When $k=0$, the difference in the pair of brackets is zero, while when $k\in \mathbb{Z}\setminus \{0\}$, each integral inside the pair of brackets is zero. Therefore we have
$$
\langle  \sinc_a, (\mathcal{H}\sinc_a)(\cdot-2k\pi)\rangle =0
$$
for any $k\in \mathbb{Z}$.

\end{proof}

\section{Bedrosian subspace of the Hilbert transform}

In this section,  we pay attention to the  set
\begin{equation}\label{eqn:subspace-S}
 \mathcal{S}_{a}:=\{f: f\in L^2(\mathbb{R}),\
\mathcal{H}(f \cos {\theta_a}(\cdot) )=f \sin {\theta_a}(\cdot )   \}.
\end{equation}
It is clear that the set $\mathcal{S}_{ a}$ is a subspace of
$L^2(\mathbb{R})$ due to the linearity of the Hilbert transform.  The subspace $\mathcal{S}_a$ shall be called the {\em Bedrosian subspace} of the Hilbert transform.  We first make a simple observation that we shall need frequently later.
\begin{lemma}\label{lem:anal}
Equation \eqref{eq:bedrosian-phase-Qian} is true if and only if for $t\in \mathbb{R}$,
\begin{equation}\label{eqn:analyticsignal0}
{\cal
H}\left(\rho(\cdot)\e^{\im\phi\cdot}\right)(t)=-\im\rho(t)\e^{\im\phi
t}.
\end{equation}
\end{lemma}
\begin{proof}
Applying the Hilbert transform $\mathcal{H}$ to both sides of equation \eqref{eq:bedrosian-phase-Qian} and  utilizing  equation \eqref{eqn:Hil1}  yields
\begin{equation}\label{eqn:hilberteqn}
\mathcal{H}( \rho(\cdot)\sin(\phi \cdot))(t)=-\rho(t)\cos(\phi t).
\end{equation}
Combining equations
(\ref{eq:bedrosian-phase-Qian}) and \eqref{eqn:hilberteqn} produces equation \eqref{eqn:analyticsignal0}.
Thus equation \eqref{eq:bedrosian-phase-Qian} is equivalent to equation \eqref{eqn:analyticsignal0}.
\end{proof}

Lemmas  \ref{lemma:fourier-transform-product}--\ref{lem:fouriere} are  several technical lemmas we need in the sequel.
\begin{lemma}\label{lemma:fourier-transform-product}
Suppose that $f\in L^2(\mathbb R)$ and $g\in L^2(-\frac{\tau}{
2},\frac{\tau}{ 2})$ is a $\tau$-periodic function having Fourier series
$g(t)=\displaystyle\sum_{k\in\mathbb Z}c_k\e^{\im k\frac{2\pi}{\tau}t}$
with $c_k =\frac{1}{\tau}\displaystyle\int_{\left(-\frac{\tau}{
2}, \frac{\tau}{2}\right)}g(t)\e^{-\im\frac{2\pi}{\tau}kt}\dd{t}$. Moreover,  the series
$\displaystyle\sum_{k\in\mathbb Z} c_k$ is
absolutely convergent. Then $f g\in L^2(\mathbb R)$ and its Fourier
transform at $\xi \in \mathbb{R}$ is given by
$$
\left(f g\right)^\wedge(\xi)=\sum_{k\in\mathbb Z}c_k\hat
f(\xi-\frac{2\pi}{\tau}k).
$$
\end{lemma}
\begin{proof}
The fact that $g\in L^2(-\frac{\tau}{
2},\frac{\tau}{2})$ implies that $|g(t)|\le c $, for a.e. $t\in \mathbb{R}$ and some constant $c$. Thus $\int_{\mathbb{R}}f^2(t)g^2(t)\dd{t}\le c^2\int_{\mathbb{R}}f^2(t)\dd{t}<\infty$. That is, $f g\in L^2(\mathbb R)$.   By the definition \eqref{eqn:fourier} of the Fourier transform, we have
\begin{eqnarray*}
( fg )^\wedge(\xi)&=&\int_{\mathbb{R}} f(t)\left(\sum_{k\in\mathbb{Z}}c_k \e^{\im k \frac{2\pi}{\tau}t } \right)\e^{-\im t\xi} \dd{t}\\
&= & \sum_{k\in\mathbb{Z}}c_k\int_{\mathbb{R}} f(t)\e^{-\im t(\xi - k \frac{2\pi}{\tau}  )} \dd{t}\\
&=& \sum_{k\in\mathbb Z}c_k\hat
f(\xi-\frac{2\pi}{\tau}k),
\end{eqnarray*}
where the interchange of the integral and the infinite sum is guaranteed by the absolute convergence of the series $\displaystyle\sum_{k\in\mathbb Z} c_k$.
\end{proof}

We say that a complex $\tau$-periodic signal $g$ is {\em circularly analytic} if its
Fourier expansion is the one-sided series
$g(t)=\displaystyle\sum_{k\in \mathbb{Z}_+}c_k\e^{\im\frac{2\pi}{\tau}kt},$ $t\in \mathbb{R}$.
\begin{lemma}\label{lemma:equivalent-conditions-real-vs-circularanalytic-realcoefficients}
Suppose that $f\in L^2(\mathbb R)$ is a real-valued function and the
complex function $g$ is $\tau$-periodic and circularly analytic with real
Fourier coefficients $c_k$, $k\in \mathbb{Z}_+$. Then
\begin{equation}\label{eq:product-analytic-real-vs-cirularanalytic}
{\cal H}(fg)(t)=-\im f(t)g(t)
\end{equation} if and only if for  $\mathrm{a.e.}$ $ \xi\in \mathbb{R}_{-}$,
\begin{equation}\label{eq:equivalent-real-vs-circularanalytic-1}
\sum _{k\in \mathbb{Z}_+} c_k\hat f(\xi-\frac{2\pi}{\tau}k)=0.
 \end{equation}
\end{lemma}
\begin{proof}
Equation (\ref{eq:product-analytic-real-vs-cirularanalytic})
indicates that $fg$ is an analytic signal. This is equivalent to say
that $fg$ has no negative frequency, that is,
$(fg)^\wedge(\xi)=0,\quad \xi\in \mathbb{R}_{-}.$ By Lemma
\ref{lemma:fourier-transform-product} and the fact that $g$ is circularly analytic, we obtain equation \eqref{eq:equivalent-real-vs-circularanalytic-1}.
\end{proof}

\begin{lemma}\label{lem:fouriere}
The $2\pi$-periodic function $g:=\e^{\im\cta(\cdot)}$ is   circularly analytic and the Fourier expansion of $ \e^{\im\cta(\cdot)}$ at $t\in\mathbb R$ is given by
\begin{equation}\label{eq:Fourier-expansion-theta-a}
\e^{\im\cta(t)}=-a+\sum _{k\in \mathbb{N}}(1-a^2)a^{k-1}\e^{\im kt}
.
\end{equation}
\end{lemma}
\begin{proof}
Recalling equation \eqref{eqn:phase}, the Fourier coefficients $c_k$, $
k\in\mathbb Z $, of $g$ is given by
$$
c_k=\frac{1}{2\pi}\int_{(-\pi,\pi)}\frac{\e^{\im t}-a}{1-a\e^{\im t}}\e^{-\im kt}\dd{t}.
$$
Let $\mathrm{Rez}(f(z);c)$ indicate the residue of the function $f$ at the point $c\in \mathbb{C}$. Denote the boundary of the unit disc $\Delta$ by $\partial \Delta$.
For $k=0$, we have $c_0=\frac{1}{2\pi \im }\oint_{\partial \Delta}\frac{z-a}{1-az}\frac{\dd{z}}{z}=\mbox{Res}\left(\frac{z-a}{(1-az)z};0\right)=-a.$
In the case $k$ is a negative integer, by using the Cauchy theorem we get that
\begin{eqnarray*}
c_k =
\frac{1}{2\pi}\int_{(-\pi,\pi)}\frac{\e^{\im t}-a}{1-a\e^{\im t}}\e^{\im \vert
k\vert t}\dd{t}=\frac{1}{2\pi \im }\oint_{\partial \Delta}\frac{(z-a)z^{|k|-1}}{1-az}\dd{z}=0,
\end{eqnarray*}
because the integrand is analytic in the unit disk $  \Delta $ as $|a|<1$.
We are left to consider the case of $k\in \mathbb{N}.$  Using the residue theorem and
the formula for $j\in \mathbb{N}$,
$$\frac{\dd{}^j}{\dd{z}^j}\left(\frac{z-a}{1-az}\right)=(1-a^2)a^{j-1}j!(1-az)^{-j-1},\quad $$
leads to when $k\in \mathbb{N}$,
\begin{eqnarray*}
c_k&=&\frac{1}{2\pi \im }\oint_{\partial \Delta}\frac{z-a}{(1-az)z^{k+1}}\dd{z}=\mbox{Res}\left(\frac{z-a}{(1-az)z^{k+1}};0\right)\\
&=& \frac{1}{k!}\lim_{z\to
0}\frac{\dd{}^k}{\dd{z}^k}\left(\frac{z-a}{1-az}\right) =(1-a^2)a^{k-1}.
\end{eqnarray*}
The proof of this lemma is completed.
\end{proof}

Next proposition indicates that the subspace $\mathcal{S}_a$ is {\em invariant} under the Hilbert transform.
\begin{prop} \label{prop:Hf}
Suppose that $\rho\in \mathcal{S}_{a}$, then $\mathcal{H}\rho \in
\mathcal{S}_{a} $.
\end{prop}
\begin{proof}
We first claim that the complex-valued signal
$(\rho+\mathrm{i}\mathcal{H}\rho)\mathrm{e}^{\mathrm{i}{\theta_a}}$ is
an analytic signal. Indeed, by Lemma \ref{lem:fouriere}   the signal $\e^{\im \theta_a} $  is circularly analytic, and we write  $\e^{\im \theta_a (t)}=\sum_{k\in \mathbb{Z}_+} g_k \e^{\im k t}$. Then by Lemma \ref{lemma:fourier-transform-product}, we have for $\xi\in \mathbb{R}$,
\begin{eqnarray}
\mathcal{F}\left( (\rho+\mathrm{i}\mathcal{H}\rho)\mathrm{e}^{\mathrm{i}{\theta_a}} \right)(\xi)&=& \sum_{k\in \mathbb{Z}_+} g_k \mathcal{F}\left(\rho+\im \mathcal{H}\rho \right)(\xi -k) . \label{eqn:fouriersinca}
\end{eqnarray}
Note that $ \left(\rho+\im \mathcal{H}\rho \right)$ is an analytic signal thus it has no negative spectrum, that is $\mathcal{F}\left(\rho+\im \mathcal{H}\rho \right)(\xi)=0$ when $\xi\in \mathbb{R}_{-}$.
Consequently,
when $\xi\in \mathbb{R}_{-}$,
$$
\sum_{k\in \mathbb{Z}_+} g_k \mathcal{F}\left(\rho+\im \mathcal{H}\rho \right)(\xi -k) =0,
$$
which implies
by Lemma \ref{lemma:equivalent-conditions-real-vs-circularanalytic-realcoefficients} that
$$
\mathcal{H}\left[ (\rho+\mathrm{i}
\mathcal{H}\rho)\mathrm{e}^{\mathrm{i}{\theta_a}} \right] =-\mathrm{i}
(\rho+\mathrm{i}\mathcal{H}\rho)\mathrm{e}^{\mathrm{i}{\theta_a}}.
$$
Simplifying both sides leads to
$\mathcal H [(\mathcal H \rho)
\mathrm{e}^{\mathrm{i}{\theta_a}}]=-\mathrm{i}[(\mathcal H \rho)
\mathrm{e}^{\mathrm{i}{\theta_a}}]$, which indicates that  $\mathcal{H}\rho\in \mathcal{S}_a$.
\end{proof}

The next lemma indicates that any real-valued function $\rho\in \mathcal{S}_a$   is completely
determined by its spectrum  $\hat\rho\chi_{(-1,1)}$ on the interval $(-1,1)$.
\begin{lemma}\label{lemma:characterization-1}
The real-valued signal $\rho\in \mathcal{S}_a$ if and only if for $\mathrm{a.e.}$ $\xi\in\mathbb
R\setminus\{0\}$, $ n\in\mathbb Z_+$,
\begin{equation}\label{eq:equivalent-conditions-1}
\hat\rho\left(\xi+ \mathrm{sgn}(\xi) n \right)=a^{
n }\hat\rho(\xi).
\end{equation}
\end{lemma}
\begin{proof}
By Lemma \ref{lem:anal} $\rho\in \mathcal{S}_a$ is equivalent to for $t\in \mathbb{R}$,
\begin{equation}
\mathcal{H}\left(\rho(\cdot)\e^{\im \theta_a(\cdot)}  \right)(t)=-\im \rho(t)\e^{\im \theta_a(t)},
\end{equation}
which is equivalent to saying that  the signal $\rho(\cdot)\e^{\im \theta_a(\cdot)}  $ is an analytic signal.
Applying Lemmas
\ref{lemma:equivalent-conditions-real-vs-circularanalytic-realcoefficients} and  \ref{lem:fouriere} to the signal $\rho(\cdot)\e^{\im \theta_a(\cdot)}  $
yields that $\rho\in \mathcal{S}_a$ holds if and only if for $
\mathrm{a.e.}$ $\xi\in \mathbb{R}_{-}$,
\begin{equation}\label{eq:amplitude-condition-negative}
a\hat\rho(\xi)=(1-a^2)\sum_{k\in \mathbb{N}} a^{k-1}\hat\rho(\xi-k).
\end{equation}

 In equation \eqref{eq:amplitude-condition-negative}, separating the first term from the rest in the sum on the right-hand side,  we have
\begin{eqnarray}
a\hat\rho(\xi)&=& (1-a^2)\left( \hat{\rho}(\xi-1)+ \sum_{k\in 1+\mathbb{N}} a^{k-1}\hat\rho(\xi-k)  \right) \nonumber\\
&=& \hat{\rho}(\xi-1)-a^2 \hat{\rho}(\xi-1)+ (1-a^2)\sum_{k\in 1+\mathbb{N}} a^{k-1}\hat\rho(\xi-k) . \label{eqn:f1}
\end{eqnarray}
 On the other hand,   multiplying both sides of equation \eqref{eq:amplitude-condition-negative} by $a$ and replacing $\xi$ by $\xi-1$ we also have
\begin{eqnarray}
a^2 \hat\rho(\xi-1)&=& (1-a^2) \sum_{k\in  \mathbb{N}} a^{k}\hat\rho(\xi-1-k) \nonumber \\
&=& (1-a^2) \sum_{k\in 1+ \mathbb{N}} a^{k-1}\hat\rho(\xi-k) . \label{eqn:f2}
\end{eqnarray}

Combining equations \eqref{eqn:f1} and \eqref{eqn:f2} deduces that for $\mathrm{a.e.}$ $ \xi\in \mathbb{R}_{-}$,
$$
\hat\rho(\xi-1)=a\hat\rho(\xi).
$$
By induction, we can get that for $n\in{\mathbb Z}_+$ and  $ \mathrm{a.e.}$ $\xi\in \mathbb{R}_{-}$,
\begin{equation}\label{eq:condition-shift-negative}
\hat\rho(\xi-n)=a^n\hat\rho(\xi).
\end{equation}
Reversing the above calculations reveals  that   condition
(\ref{eq:condition-shift-negative}) is also sufficient for
(\ref{eq:amplitude-condition-negative}).

Taking the conjugate of  both sides of    equation
(\ref{eq:amplitude-condition-negative}) and using the Hermitian property of $\hat{\rho}$, we obtain  an equivalent equation
\begin{equation}\label{eq:amplitude-condition-positive}
a\hat\rho(\xi)=(1-a^2)\sum _{k\in \mathbb{N}} a^{k-1}\hat\rho(\xi+k)
\end{equation}
 for $\mathrm{a.e.}$ $ \xi\in \mathbb{R}_{-}$.
A similar discussion  to the case $\xi\in \mathbb{R}_{-}$ gives that equation
(\ref{eq:amplitude-condition-positive}) is equivalent to
\begin{equation}\label{eq:condition-shift-positive}
\hat\rho(\xi+n)=a^n\hat\rho(\xi)
\end{equation}
for $n\in{\mathbb Z}_+$ and  $ \mathrm{a.e.}$ $\xi\in \mathbb{R}_{-}$.
Finally combining equations (\ref{eq:condition-shift-negative}) and
(\ref{eq:condition-shift-positive}) produces
(\ref{eq:equivalent-conditions-1}).
\end{proof}

\begin{prop}\label{lem:Sinc_bedrosian}
$$  \sinc_a\in \mathcal{S}_a\quad \mathrm{and}\quad \mathcal{H}\sinc_a \in \mathcal{S}_a.$$
\end{prop}
\begin{proof}
By Lemma \ref{prop:Hf}, it suffices to show $\sinc_a\in \mathcal{S}_a$.
Observing the identity $H_a(\cdot+n)=a^{\vert n\vert}H_a(\cdot)$  for a.e.  $\xi\in \mathbb{R}\setminus \{0\}$, $n\in \mathbb{Z}$ and recall equation \eqref{general-sinc-function-f}  we thus have
\begin{eqnarray*}
\mathcal{F}(\sinc_a)(\xi+\vert n\vert \mathrm{sgn}(\xi))&=&a^{\vert
n\vert}\mathcal{F}(\sinc_a)(\xi).
\end{eqnarray*}
Hence by Lemma \ref{lemma:characterization-1}, we conclude that $\sinc_a\in \mathcal{S}_a$.

\end{proof}

\section{An observation--how a linear phase determines the amplitude}
\setcounter{equation}{0}

In this section, we specifically consider the case when the phase $\phi$ in equation (\ref{eq:bedrosian-phase-Qian}) is a linear phase and investigate the representation of the corresponding amplitude. The  result is given in the following theorem. We use the notation $\mathrm{supp}(f)$ for the set of real numbers on which the values $f(x)$ at $x\in \mathbb{R}$ are nonzero.

\begin{theorem}\label{th-linearcase}
Suppose that $\gamma$ is a positive real number and $\rho$ is a non-zero real
signal in $L^2(\mathbb R)$. Then the following equation
\begin{equation}\label{eq:bedrosian-linear-phase}
{\cal
H}\left(\rho(\cdot)\cos(\gamma\cdot)\right)(t)=\rho(t)\sin(\gamma t)
\end{equation}
holds if and only if $\rho$ is bandlimited with
$\mathrm{supp}({\hat\rho})\subset [-\gamma,\gamma].$
\end{theorem}

\begin{proof}
 Lemma \ref{lem:anal} implies that  equation \eqref{eq:bedrosian-linear-phase} is equivalent to \begin{equation}\label{eqn:analyticsignal}
{\cal
H}\left(\rho(\cdot)\e^{\im\gamma\cdot}\right)(t)=-\im\rho(t)\e^{\im\gamma
t},
\end{equation}
or equivalently,  the
complex signal $\rho(t)\e^{\im\gamma t}$ is an analytic signal.
Note that the Fourier transform of the function  on the left-hand side of equation \eqref{eqn:analyticsignal} is given by
$-\im\; {\mathrm{sgn}}(\xi)\hat\rho(\xi-\gamma)$ and the Fourier transform of the right-hand side of equation \eqref{eqn:analyticsignal} is  $-\im\hat\rho(\xi-\gamma)$.  This leads to the equivalent equation of (\ref{eq:bedrosian-linear-phase}) in the  frequency
domain
\begin{equation}
\left(\mathrm{sgn}(\xi)-1\right)\hat\rho(\xi-\gamma)=0,
\end{equation}
i.e.
\begin{equation}\label{eq:bedrosian-linear-phase-frequency}
\hat\rho(\xi)=0, \quad\mathrm{for}\quad \xi\in(-\infty,-\gamma].
\end{equation}
Thus if $\rho$ is bandlimited with $\mathrm{supp} \hat{\rho} \subset [\gamma, \gamma]$, then clearly   equation (\ref{eq:bedrosian-linear-phase-frequency}) is true, hence equation \eqref{eq:bedrosian-linear-phase} is true.  On the other hand,   assuming that equation \eqref{eq:bedrosian-linear-phase} is true,
we obtain that $\hat\rho(\xi)=0$ for $\xi\in(-\infty,-\gamma].$
Since $\rho$ is real-valued, by the Hermitian property of the Fourier transform of a real signal,   $\hat\rho(\xi)=0$ for
$\xi\in[\gamma,+\infty).$ Consequently we conclude that  $\rho$ is a
bandlimited function with its support in the frequency domain belongs to
$[-\gamma,\gamma].$
\end{proof}

\noindent  Remark: From this theorem we know that the amplitude function $\rho$ can be represented by  shifts of sinc function if and only if equation \eqref{eq:bedrosian-linear-phase} is true.
\section{How does non-linear phase determine amplitude?}
\setcounter{equation}{0}

In this section, we shall completely characterize a real-valued function $\rho\in \mathcal{S}_a$. We begin with  introducing  two one-sided filters that are related to the two-sided symmetric cascade filter $H_a$:
\begin{equation}\label{eq:laddershape-filter-onesides}
H^+_a(t):=H_a(t)\chi_{{\mathbb R}_+}(t),\qquad
H^-_a(t):=H_a(t)\chi_{{\mathbb R}_-}(t),
\end{equation}
where $ t\in\mathbb R$.
The next   lemma provides the Fourier transform pairs of $H_a^+$ and $H_a^-$, respectively.
\begin{lemma}\label{lem:fourierHa}
 The Fourier transform of $H^+_a$ and $H_a^-$ are  given, respectively, by
\begin{eqnarray}
(H^+_a(\cdot))^\wedge(\xi) &=& \frac{1}{\sqrt{2\pi}}\frac{1}{1-a\e^{-\im \xi}}\frac{1-\e^{-\im \xi}}{\im \xi}, \label{eqn:fourierHa}
\end{eqnarray}
and
$$
(H^-_a(\cdot))^\wedge(\xi)=\frac{1}
{\sqrt{2\pi}}\frac{1}{1-a\e^{\im \xi}}\frac{1-\e^{\im \xi}}{-\im \xi}.
$$
\end{lemma}
 \begin{proof}
 We start with the  observation
 $$
 H_a^+(t)=\sum_{k\in \mathbb{Z}_+}a^k\chi_{[k,k+1)}(t)=(1-a)\sum_{k\in \mathbb{N}} a^{k-1}\chi_{[0,k)}(t).
 $$
 By the definition of the Fourier transform \eqref{eqn:fourier}, we obtain that
\begin{eqnarray*}
(H^+_a(\cdot))^\wedge(\xi) &=& \frac{1-a}{\sqrt{2\pi}}
\int_{\mathbb{R}}\sum_{k\in \mathbb{N}}a^{k-1}\chi_{[0,k)}(t)\e^{-\im \xi t}\dd{t}\\
&=&\frac{1-a}{\sqrt{2\pi}}\sum_{k\in \mathbb{N}}a^{k-1}\int _{[0,k)}\e^{-\im \xi t}\dd{t}\\
&=&\frac{1-a}{\sqrt{2\pi}}\sum_{k\in \mathbb{N}}a^{k-1}\frac{\e^{-\im k\xi}-1}{-\im \xi},
\end{eqnarray*}
where, the interchange of the order of the summation and the integral is justified by the absolute convergence of the series. Continue by noting $\sum_{k\in \mathbb{N}} a^{k-1}\e^{\im k \xi}=\frac{\e^{\im \xi}}{1-a\e^{\im \xi}} $ we have
\begin{eqnarray*}
(H^+_a(\cdot))^\wedge(\xi) &=&\frac{1}{\sqrt{2\pi}}\frac{1-a\e^{-\im \xi}-(1-a)\e^{-\im \xi}}{\im \xi(1-a\e^{-\im \xi})}\\
&=&\frac{1}{\sqrt{2\pi}}\frac{1}{1-a\e^{-\im \xi}}\frac{1-\e^{-\im \xi}}{\im \xi}.
\end{eqnarray*}

 Using the identity
$H^-_a(\cdot)=H^+_a(-\cdot)$, we obtain that
$$
(H^-_a(\cdot))^\wedge(\xi)=(H^+_a(\cdot))^\wedge(-\xi)=\frac{1}
{\sqrt{2\pi}}\frac{1}{1-a\e^{\im \xi}}\frac{1-\e^{\im \xi}}{-\im \xi}.
$$
\end{proof}

\begin{lemma}\label{eqn:lem:be}
The   quasi-Bedrosian type identity
\begin{equation}\label{eqn:quasibe}
{\cal
H}\left(p_a(\cdot)\sinc(\cdot)\right)(t)=\frac{1+a}{1-a}p_a(t){\cal
H}\sinc(t)
\end{equation}
is true for $ t\in\mathbb R$ and
\begin{equation}\label{eqn:sincHsinc}
\frac{1}{1-ae^{\im t}}\frac{\e^{\im t}-1}{\im t}=\frac{1}{1+a}\left(\sinc_a(t)+\im {\cal
H}\sinc_a(t)\right).
\end{equation}
\end{lemma}
\begin{proof}
Let $r(t):=\frac{1}{1-ae^{\im t}}\frac{\e^{\im t}-1}{\im t}$. Separating the real part from the imaginary part of $r(t)$ produces  \begin{eqnarray}
r(t) &=&\frac{1-a}{1-2a\cos
t+a^2}\frac{\sin t}{t}+\im  \frac{1+a}{1-2a\cos t+a^2}\frac{1-\cos
t}{t} \nonumber \\
&=& \frac{1}{1+a}p_a(t)\sinc(t)+\im \frac{1}{1-a}p_a(t){\cal
H}\sinc(t),\label{eqn:r(t)1}
\end{eqnarray}
where we have used the identity
\begin{equation}\label{eqn:hsinc}
\mathcal{H}\sinc (t) =\frac{1-\cos t}{t} .
\end{equation}

On the other hand, observe that  $r(t)=\sqrt{2\pi} \left( H_a^+ \right)^\wedge (-t)$, so the Fourier transform $\hat{r}(t)=\sqrt{2\pi} H_a^+(t)$   has zero negative spectrum. This implies $r$ is an analytic signal. Therefore the imaginary part of $r(t)$ equals to the Hilbert transform of its real part, that is, $$
\mathcal{H} \left( \frac{1}{1+a}p_a(\cdot)\sinc(\cdot)  \right)(t) =\frac{1}{1-a}p_a(t){\cal
H}\sinc(t).
$$
After rearranging the above equation we obtain equation \eqref{eqn:quasibe}. Now we turn to show equation \eqref{eqn:sincHsinc}.  Applying equation \eqref{eqn:quasibe} to the imaginary part of  equation \eqref{eqn:r(t)1} we continue to have that
\begin{eqnarray*}
r(t)&=&   \frac{1}{1+a}\left(p_a(t)\sinc(t)+ \im {\cal
H}(p_a(\cdot)\sinc(\cdot))\right)(t)\\
&=& \frac{1}{1+a}\left(\sinc_a(t)+\im {\cal
H}\sinc_a(t)\right),
\end{eqnarray*}
which is equation \eqref{eqn:sincHsinc}.
\end{proof}
\begin{coro} \label{coro:cosinc}
Let $\mathrm{cosinc}_a(t):=\frac{1-\cos\theta_a(t)}{t}$ and define $\mathrm{cosinc}_a(0)=\lim_{t\to 0}\frac{1-\cos\theta_a(t)}{t}$, then for $t\in \mathbb{R}$,
\begin{equation}\label{eqn:Hsina}
\mathcal{H}\sinc_a(t)=\mathrm{cosinc}_a(t).
\end{equation}
Moreover, the function $\mathrm{cosinc}_a$ is odd, bounded, infinitely differentiable,  $\mathrm{cosinc}_a(t)\le \left(\frac{1+a}{1-a}\right)^2\frac{3}{|t|+1}$ for $t\in \mathbb{R}$, and $\mathrm{cosinc}_a\in L^2(\mathbb{R})$.
\end{coro}
\begin{proof}
Equation \eqref{eqn:Hsina} is just a new form of equation \eqref{eqn:quasibe}. Indeed, the left-hand side of equation \eqref{eqn:quasibe} is just $\mathcal{H}\sinc_a (t)$. Utilizing equations \eqref{eqn:pa} and \eqref{eqn:hsinc} then   the right-hand side of equation \eqref{eqn:quasibe} is   simplified to  $\mathrm{cosinc}_a(t)$. The second statement follows by rewriting $\mathrm{cosinc}_a$ as
\begin{equation}\label{eqn:cosinca}
\mathrm{cosinc}_a(t)=\frac{1+a}{1-a}p_a(t)\frac{1-\cos t}{t},
\end{equation}
and noting the function
\begin{equation}\label{eqn:hsinc}
\mathcal{H}\sinc (t)=\frac{1-\cos t}{t}
 \end{equation}
 is infinitely differentiable if we define $\mathcal{H}\sinc (0)=\lim_{t\to 0}\frac{1-\cos t}{t}$, and $\left|\frac{1-\cos t}{t}  \right|\le \frac{3}{1+|t|}$ for any $t\in \mathbb{R}$.
\end{proof}

An example of the graph of $\mathrm{cosinc}_a$ with $a=0.5$ is shown in figure \ref{fig:cosinc}. As a comparison, the graph of $\mathcal{H}\sinc (t) =\frac{1-\cos t}{t}$ is also shown in the figure.

\begin{figure}[ht]
$$
\includegraphics[scale=0.6]{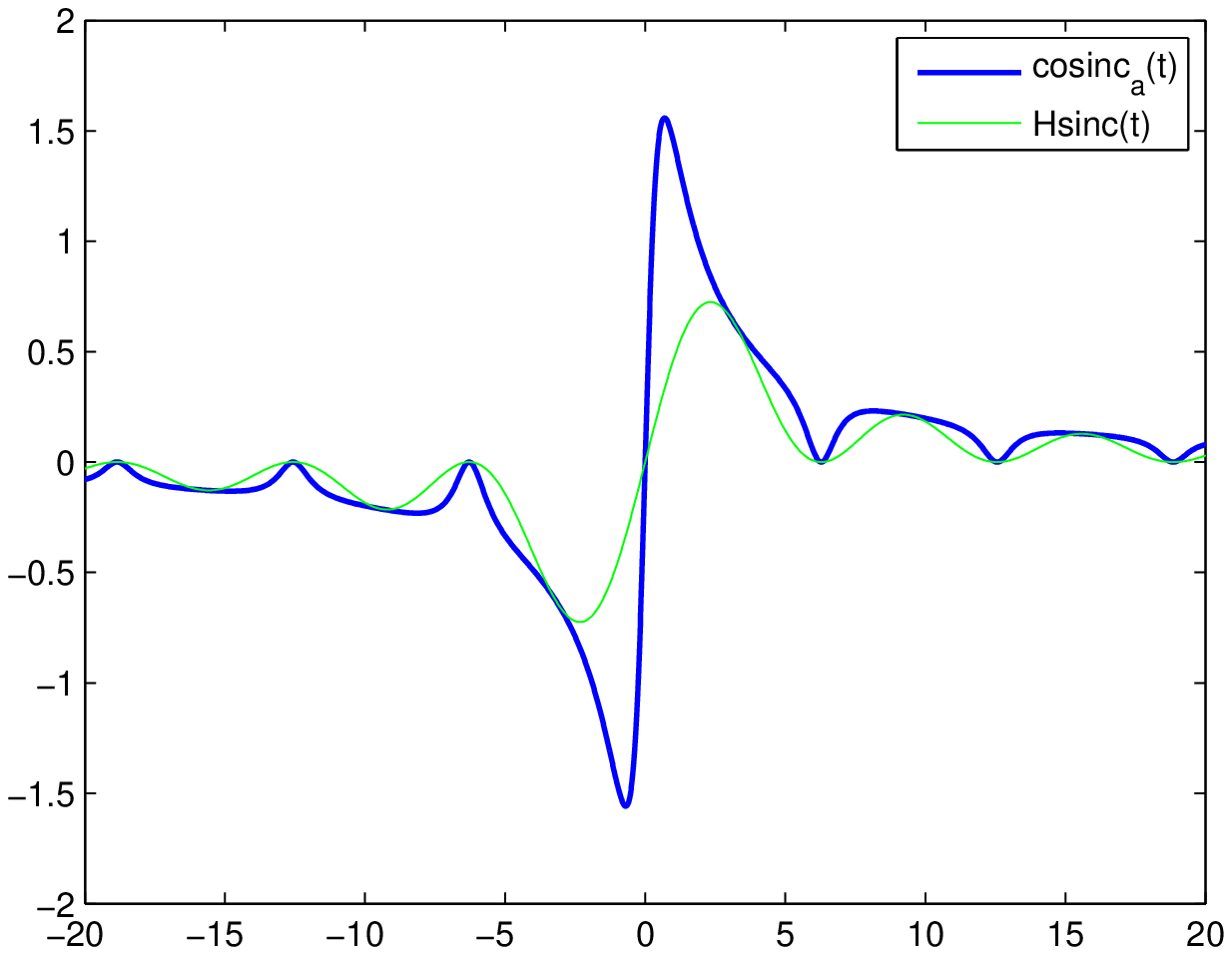}
$$
\caption{The graph of $\mathrm{cosinc}_a$,   $a=0.5$.}\label{fig:cosinc}
\end{figure}
The following lemma appeared in \cite{CCK10}. However,   a completely new  proof  by direct construction is given here for both sufficiency and necessity. The construction approach   casts a new insight to understand the relation of the spectrum $\hat{\rho}$ of a function $\rho\in \mathcal{S}_a$ to the symmetric cascade filter $H_a$.
\begin{lemma}\label{th:characterization}
A real signal $\rho\in \mathcal{S}_a$ %
if and only if there are two real
sequences $\boldsymbol{r}=\{r_k:k\in\mathbb Z\} $ and $\boldsymbol{s}=\{s_k:k\in\mathbb Z\}$ in $ l^2(\mathbb{Z})$
such that for $  t\in\mathbb R$,
\begin{equation}\label{eq:characterization}
\rho(t)=\sum_{k\in\mathbb
Z}r_k\sinc_a(t-2k\pi)+\sum_{k\in\mathbb Z}s_k({\cal
H}\sinc_a)(t-2k\pi).
\end{equation}
\end{lemma}

\begin{proof}
We first show the necessity, that is, if $\rho\in \mathcal{S}_a$, then $\rho$ is given by equation \eqref{eq:characterization}. From Lemma
\ref{lemma:characterization-1}, we know that
$\rho(\xi)\chi_{\mathbb{R}_+}(\xi)$ and $\rho(\xi)\chi_{\mathbb{R}_-}(\xi)$ are
completely determined by $\rho(\xi)\chi_{(0,1)}(\xi)$ and
$\rho(\xi)\chi_{(-1,0)}(\xi)$, respectively. Expand
$\rho(\xi)\chi_{(0,1)}(\xi)$ and $\rho(\xi)\chi_{(-1,0)}(\xi)$ into
1-periodic functions $\gamma_1(\xi):=\displaystyle\sum_{k\in\mathbb
Z}u_k \e^{\im 2k\pi t}$ and
$\gamma_2(\xi):=\displaystyle\sum_{k\in\mathbb Z}v_k \e^{\im 2k\pi t}$, where the two series $(u_k: k\in \mathbb{Z})$ and $(v_k: k\in \mathbb{Z})$ are in $l^2(\mathbb{Z})$.

By Lemma \ref{lemma:characterization-1} we thus obtain that for $
\mathrm{a.e.}$ $\xi\in\mathbb R$,
$$
\hat\rho(\xi)=\gamma_1(\xi)H^+_a(\xi)+\gamma_2(\xi)H^-_a(\xi).
$$
The Hermitian property of $\hat\rho$ implies that for $  \xi\in\mathbb
R\setminus\{0\}$,
$$
\gamma_2(\xi)= \left(\gamma_1(-\xi)\right)^* .
$$
For convenience let $\gamma=\gamma_1$. Consequently $\hat\rho$ has the representation
$$
\hat\rho(\xi)=\gamma(\xi)H^+_a(\xi)+\left(\gamma(-\xi)\right)^*H^-_a(\xi),
$$
for $
\mathrm{a.e.}$ $\xi\in\mathbb R $.
Let $\rho_A$ be the analytic signal associated with $\rho$, which may be defined by
$$\hat\rho_A(\xi)=2\hat\rho(\xi)\chi_{\mathbb R_+}(\xi)$$
or equivalently for $  t\in\mathbb R$,
$$
\rho_A(t)=\rho(t)+\im{\cal H}\rho(t).
$$
Then we have $\rho(t)=\mbox{Re}(\rho_A(t)).$ We need to investigate
$\rho_A$. The inverse Fourier transform of $\hat{\rho}_A$ yields
\begin{eqnarray*}
\rho_A(t)&=&\frac{1}{\sqrt{2\pi}}\int_{\mathbb{R}}
2\gamma(\xi)H^+_a(\xi)\e^{\im \xi t}\dd{\xi}\\
&=& \frac{2}{\sqrt{2\pi}}\int_{\mathbb{R}}\sum_{k\in\mathbb
Z}u_k \e^{-\im 2k\pi\xi}H^+_a(\xi)\e^{\im \xi t}\dd{\xi}.
\end{eqnarray*}
Applying the Lebesgue dominated convergence theorem to interchange the order of the integral and the sum, and appealing to equation  \eqref{eqn:fourierHa}
we obtain that
\begin{eqnarray*}
\rho_A(t)&=&\frac{2}{\sqrt{2\pi}}\sum_{k\in\mathbb
Z}u_k\int_{\mathbb{R}}H^+_a(\xi)\e^{\im \xi(t-2k\pi)}\dd{\xi}\\
&=&\frac{2}{\sqrt{2\pi}}\sum_{k\in\mathbb Z}u_k
\frac{1}{1-a \e^{\im (t-2k\pi)}}\frac{\e^{\im (t-2k\pi)}-1}{\im (t-2k\pi)}.
\end{eqnarray*}
Denote $u_k=u^{(1)}_k+\im u^{(2)}_k$.  Recalling equation \eqref{eqn:sincHsinc}
we readily obtain that
$$
\rho(t)=\frac{2}{\sqrt{2\pi}(1+a)} \left( \sum_{k\in\mathbb
Z}u^{(1)}_k\sinc_a(t-2k\pi)- \sum_{k\in\mathbb
Z}u^{(2)}_k({\cal H}\sinc_a(\cdot))(t-2k\pi)\right).
$$
Let $r_k=\frac{2}{\sqrt{2\pi}(1+a)}u^{(1)}_k$ and
$s_k=-\frac{2}{\sqrt{2\pi}(1+a)}u^{(2)}_k$.
Since the series $(u_k,k\in \mathbb{Z})\in l^2(\mathbb{Z})$, we must have both series $(r_k,k\in \mathbb{Z})$ and $(s_k,k\in \mathbb{Z})$ in $l^2(\mathbb{Z})$.
We have arrived at
equation (\ref{eq:characterization}).

We now turn to the proof of sufficiency. It suffices to check that a
function $\rho\in L^2(\mathbb R)$ having the form
(\ref{eq:characterization}) satisfies the equation
(\ref{eq:equivalent-conditions-1}) by Lemma \ref{lemma:characterization-1}.  Applying the Fourier transform to both sides of equation  (\ref{eq:characterization}) and noting that the Fourier
transform of the generalized sampling function given by equation \eqref{general-sinc-function-f}  we get that
\begin{eqnarray*}
\hat\rho(\xi)= H_a(\xi)\left(M(\xi)-\im\ \mbox{sgn}(\xi)G(\xi)\right)
\end{eqnarray*}
with the two $1$-periodic functions
\begin{eqnarray*}
M(\xi)=\frac{\sqrt{2\pi}}{2}(1+a)\sum_{k\in \mathbb Z}r_k
\e^{-\im 2k\pi\xi},\quad G(\xi)=\frac{\sqrt{2\pi}}{2}(1+a)\sum_{k\in \mathbb
Z}s_k\e^{-\im 2k\pi\xi}.
\end{eqnarray*}
Observing the identity $H_a(\cdot+n)=a^{\vert n\vert}H_a(\cdot)$ and
the 1-periodicity of $M$ and $G$ lead to that, for   $\xi\in \mathbb{R}_+$,
\begin{eqnarray*}
\hat\rho(\xi+\vert n\vert)&=&H_a(\xi+\vert n\vert)(M(\xi+\vert
n\vert)-\im G(\xi+\vert n\vert))\\
&=&a^{\vert n\vert}H_a(\xi)\left(M(\xi)-\im G(\xi)\right)=a^{\vert
n\vert}\hat\rho(\xi).
\end{eqnarray*}
The case with negative $\xi$ can be shown similarly. This completes the proof of the lemma.
\end{proof}



\begin{theorem}
If $\rho\in  \mathcal{S}_a$, then both $\rho $ and $\mathcal{H}\rho$  are continuous.
\end{theorem}
\begin{proof}
 By Proposition \ref{prop:Hf}, $\mathcal{H}\rho \in \mathcal{S}_a$, thus it suffices to show $\rho\in \mathcal{S}_a$ is continuous.
 Indeed, if $\rho\in \mathcal{S}_a$, then $\rho$ has the representation   \eqref{eq:characterization} by Lemma \ref{th:characterization}. Since both $\sinc_a$ and $\mathcal{H}\sinc_a$ are in $L^2(\mathbb{R})$ by Proposition \ref{prop:sinc} and Corollary \ref{coro:cosinc},  consequently by the Cauchy-Schwarz inequality, the two series on the right-hand side of equation \eqref{eq:characterization} converge uniformly, hence the limiting function $\rho$ is continuous.

\end{proof}

\begin{lemma}\label{lem:ortho-coef}
If    $f\in
\mathcal{S}_{a}$,  then for $k\in \mathbb{Z}$,
$$
\langle f, \sinc_{a}(\cdot -2k\pi )\rangle = f(2k\pi),
\quad \mbox{and}\quad \langle f,
\mathcal{H}\sinc_{a}(\cdot -2k\pi) \rangle =-
\mathcal{H}f(2k\pi).
$$
\end{lemma}
\begin{proof}
The $2\pi$-periodicity of $\sin {\theta_a}$ yields
$$
 \langle f,
\sinc_{a}(\cdot -2k\pi)\rangle = \int_{\mathbb{R}} f(t)
\frac{\sin {\theta_a} (t-2k\pi)}{t-2k\pi}\dd{t} = - \mathcal{H} (f\sin
{\theta_a})(2k\pi).
$$
Therefor by the assumption on $f$, we have when  $k\in \mathbb{Z}$,
$$
- \mathcal{H} (f\sin {\theta_a})(2k\pi)=f(2k\pi)\cos {\theta_a} (2k\pi)
=f(2k\pi)
$$
because $\cos {\theta_a} (2k\pi)=\cos {\theta_a} (0)=1$ by equation \eqref{eqn:costhetaa}.

The second equality follows by noting the fact that $\mathcal{H}f \in \mathcal{S}_{a}$
 due to Proposition \ref{prop:Hf}
and the Hilbert transformer is an anti-self adjoint operator thus
$$\langle f, \mathcal{H}\sinc_{a}(\cdot -2k\pi\rangle=
-\langle \mathcal{H}f, \sinc_{a}(\cdot -2k\pi\rangle =-\mathcal{H}f(2k\pi). $$
\end{proof}

Lemma \ref{th:characterization}, Lemma \ref{lem:ortho-coef}, and Theorem \ref{prop:ortho} immediately implies the following theorem.
\begin{theorem}\label{thm:sampling}
Any real-valued
function $\rho \in \mathcal{S}_{a}$  {\em if and only
if}
$$
\rho (t)=\frac{1-a}{1+a}\sum_{k\in \mathbb{Z}}\rho(2k\pi)
 {\mathrm{sinc}}_{a} (t-2k\pi)- \frac{1-a}{1+a}\sum_{k\in
\mathbb{Z}}\mathcal{H}\rho(2k\pi)\mathcal{H}
{\mathrm{sinc}}_{a}(t-2k\pi).
$$
Moreover, the sampling sequences $\left(\rho(2k\pi), k\in
\mathbb{Z}\right) $ and
$\left(\mathcal{H}\rho(2k\pi), k\in \mathbb{Z}\right) $ are in $
l^2(\mathbb{Z})$.
\end{theorem}
From Theorem \ref{thm:sampling} and Theorem \ref{prop:ortho} we
have the following corollary.
\begin{coro}
The system $\Phi$ defined in \eqref{eqn:Phi} is a complete
orthogonal system of the subspace  $\mathcal{S}_{a}\subset L^2(\mathbb{R}) $. That is, the subspace equals to the closure of the span of the set $\Phi$, or symbolically,
$$
\mathcal{S}_{a} =\overline{\mathrm{span} \Phi}.
$$
\end{coro}

Lastly, we state some facts for the case $a=0$. Note when $a=0$, $\sin \theta_a(t)$ becomes $\sin t$, $\cos \theta_a(t)$ becomes $\cos t$,  $\sinc_a$ becomes $\sinc$ and $\mathcal{H}\sinc_a$ becomes $\mathcal{H}\sinc$ which is given at $t\in \mathbb{R}$ by equation \eqref{eqn:hsinc}. Therefore Theorem \ref{th-linearcase}, Theorem \ref{thm:sampling} and the well-known Shannon sampling theorem imply the following corollary.
\begin{coro}
The following statements are equivalent:
\begin{enumerate}
\item The real signal $\rho\in L^2(\mathbb{R})$ is bandlimited such that $\mathrm{supp}\hat{\rho}\subseteq [-1,1]$.
    \item $\mathcal{H}(\rho(\cdot) \cos(\cdot)) (t) =\rho(t) \sin t$, $t\in \mathbb{R}$.
    \item $\mathcal {H} (\rho(\cdot)\e^{\im \cdot}  )(t)=-\im \rho(t)\e^{\im t}$, $t\in \mathbb{R}$.
     \item    $
\rho (t)=\sum_{k\in \mathbb{Z}}\rho(2k\pi)
 {\mathrm{sinc}} (t-2k\pi)- \sum_{k\in
\mathbb{Z}}\mathcal{H}\rho(2k\pi)
\frac{1-\cos(t-2k\pi)}{t-2k\pi}
$, $t\in \mathbb{R}$.
\item $\rho(t)=\sum_{k\in \mathbb{Z}} \rho{(k\pi)}\sinc(t-k\pi)$, $t\in \mathbb{R}$.
\end{enumerate}
\end{coro}

\end{document}